\documentclass{ws-ijtaf}

\begin{document}

\markboth{P. Cirillo, J. H\"usler and P. Muliere }
{A nonparametric urn-based approach to interacting failing systems}

\catchline{}{}{}{}{}
%%%%%%%%%%%%%%%%%%%%%%%%%%%%%%%%%%%%%%%%%%%%%%%%%%%%%%%%%%%%%%%%%%%%

\title{A NONPARAMETRIC URN-BASED APPROACH \\TO INTERACTING FAILING SYSTEMS\\ WITH AN APPLICATION TO CREDIT RISK MODELING}

\author{PASQUALE CIRILLO\footnote{Corresponding author: pasquale.cirillo@stat.unibe.ch, tel. +41 (0)31 63 18 803.}} 
\author{J\"URG H\"USLER}
\address{Institute of Mathematical Statistics and Actuarial Sciences, University of Bern \\
Sidlerstrasse 5, Bern, CH-3012, Switzerland}
\author{PIETRO MULIERE}
\address{Department of Decision Sciences, Bocconi University \\
Via R\"ontgen 1, Milan, IT-20136, Italy}

\maketitle

\begin{history}
\received{($26^{th}$ May 2009)}
\revised{($17^{th}$ December 2009)}
\end{history}

\begin{abstract}
In this paper we propose a new nonparametric approach to interacting failing systems (FS), that is
systems whose probability of failure is not negligible in a fixed time
horizon, a typical example being firms and financial bonds.\\
The main purpose when studying a FS is to calculate the probability of
default and the distribution of the number of failures that may occur during
the observation period. A model used to study a failing system is defined 
\textit{default model}.\\
In particular, we present a general recursive model constructed by the means
of interacting urns.\\
After introducing the theoretical model and its properties we show a first application to credit risk modeling, showing how to assess the idiosyncratic probability of default of an obligor and the joint probability of failure of a set of obligors in a portfolio of risks, that are divided into reliability classes.
\end{abstract}


\section{Introduction}\label{section1}

A failing system (FS) is a system whose probability of failure is not
negligible in a fixed time horizon (for example one year). The interest for
such a topic is due to its diffusion in real-life problems: financial
portfolios and credit risk, electrical and mechanical systems, firms' defaults, the world wide web can all be considered failing systems.\\
The main purpose when studying a FS is to calculate the probability of
default and the distribution of the number of failures that may occur during
the observation period. A model that studies failing systems is defined 
\textit{default model}.\\
Here we present new ways for calculating the probabilities of joint
defaults in $k$ different homogeneous groups of FS, when each group is
characterized by some sort of external information about
its reliability, which allows for an ordering. In particular, we propose a recursive model constructed by
the means of interacting urns.\\
For simplicity, we assume that the probability of default is homogeneous within groups. Moreover we hypothesize that this probability of default is given by the sum of two components:
\begin{enumerate}
\item a group idiosyncratic probability of failing that reveals, on the average, how a FS belonging to a given group is likely to fail "on its own";
\item a systemic probability of failing, which represents the amount of negative/positive interactions among failing systems in different groups. In particular, once we have ordered the groups of FS from the best to the worst one, we will assume that the systemic probability of failing of group $i$ increases (decreases) if the number of defaults in the superior groups $1,...,i-1$ increases (decreases), while it remains the same for any change in the inferior groups $i+1,...,k$.\\
\end{enumerate}
Hence, the aim of this paper is to model the dependence among failures both
within and between the $k$ groups. This scheme can be efficiently reproduced with urns.\\
Urn processes (or urn models or urn schemes) constitute a very large family
of probabilistic models in which the probability of certain events is
represented in terms of sampling, replacing and adding balls in one or more
urns or boxes.\\
Urn problems have been an important part of the theory of probability since
the publication of the posthumous \textit{Ars conjectandi} by Jakob Bernoulli \cite{Bernoulli} in 1713. Their most interesting characteristic is the possibility
of simplifying complex probabilistic ideas, making them intuitive and
concrete, and yet guaranteeing a good level of abstraction, that allows for
general results.\\
The choice of urn processes as a probabilistic tool is mainly due to the
following reasons:

\begin{enumerate}
\item They are particularly suitable, thanks to their efficiency, for
studying chance experiments, especially when these are characterized by
countable spaces;
\item They represent an excellent way to describe the concept of
\textquotedblleft random choice\textquotedblright ;
\item Simple urns can be easily compounded into new ones in order to study
more complex problems;
\item Urn schemes have as powerful as elegant combinatorial properties, that
allow for general, complex results in a rather concise form.
\item There are many relationships and isomorphisms between urn models and
other well-known mathematical objects (see for example \textit{analytic urns} in \cite{Flajolet}). All this gives the possibility to the
researcher of switching from one approach to the other at her/his
convenience;
\item Urns are very useful objects in simulations, given their natural
connections with sampling schemes.
\end{enumerate}

As discussed in \cite{Bala} and \cite{Johnson}, one of the prototypes of urn processes
is the well-known Polya urn, developed at the beginning of the last century
to model the diffusion of infectious diseases. It represents one of the
simplest ways to generate beta exchangeable random variables and it is based
on the concept of reinforcement. Moreover, its multidimensional version, as
shown in \cite{Black}, is a very useful tool to obtain the fundamental
Dirichlet distribution, an essential tool for Bayesian statistics.\\
Polya urns represent one of the basic pillars of the present work. We will make use of Polya urns to
define our urn chain model.\\
The analysis of default models via urns has several advantages:

\begin{enumerate}
\item The modelization is rather intuitive and immediate;
\item Urns can be considered a first attempt to study failing systems from a
Bayesian nonparametric point of view, for they allow the researcher to
introduce her/his prior knowledge into the analysis by modifying the initial
urn composition and the reinforcement matrix of the urn process (see \cite{Cifa} for more details)
\item The flexibility of urn schemes is a very useful
characteristic for simulations and empirical studies, as we will see in the
last section of the paper.
\end{enumerate}
Some seminal ideas for our construction have been introduced in Marsili and Valleriani \cite{Mava}, where systems of interacting Polya urns are discussed.
Anyway, the main references for our model, which is based on
an iterative framework of interacting urns, are \cite{Doksum} and \cite{Mulwal}.
In particular, the basic ideas in \cite{Mulwal}, further developed in \cite{Cirillo}, can be considered the very
starting point of this work.\\
The general framework we propose in this paper has, according to us, several interesting applications; a useful one being related to credit risk modeling.\\
Credit risk is the risk of loss due to a debtor's non-payment of a loan or other line of credit and it is strictly linked to the concept of default. One of the most important issues in credit risk modeling is represented by the assessment of the probability of failure/default of an obligor and/or a set of obligors in a portfolio of risks. In particular, especially for banks and financial companies, it is important to estimate the probability of joint defaults over a fixed time horizon, and this is why default models represent a fundamental tool in credit risk analysis. Furthermore, the assessment of the probability of joint defaults, as pointed out in \cite{Schoen}, is also important for securities whose payoff is function of the profits and losses of a portfolio of underlying bonds. It follows that the study of dependence structure of interacting failing systems such as firms and bonds is fundamental for a correct estimation of credit risk, especially when the dependence cannot be summarized by simple measures of co-variability like linear correlation, despite this over-simplification is often used in credit risk modeling (see for example, \cite{McNeil}).\\
The paper is organized as follows: Section 2 introduces our model and the main probabilistic results; Section 3 propose a first study concerning credit risk modeling; and Section 4 concludes.

\section{Introducing the urn chain model}
We can now introduce the urn chain model in which several urns interact to reproduce dependent risks, the basic brick of our construction being represented by Polya urn.\\
Polya urn has been introduced by Polya and Eggenberger in 1923 \cite{Polya} to study infectious and self-reinforcing phenomena. The most important characteristic of this particular urn scheme is its reinforcement mechanism, that has become the prototype for many probabilistic models for studying contagion and aftereffects. The behavior of Polya urn is very simple yet ingenious. In its simplest two-color version, imagine we have an urn containing balls of two different colors (say black and white). Every time we sample the urn we look at the color of the chosen ball and then put it back into the urn together with another ball of the same color. In this way, the more a given color has been sampled in the past, the more likely it will be sampled in the future. Obviously the reinforcement rule can be generalized introducing $s$ balls of the same colors, considering a random or time-varying reinforcement and so on. For a complete analysis of Polya urn's combinatorics, behavior  and generalizations see \cite{Johnson} and \cite{Mahmoud}.\\
In what follows, we will make use of simple two-color Polya urns with general (but fixed) reinforcement $s$. Our choice of Polya urns to construct a risk model - with possible applications to credit risk - has been inspired by several works available in the literature, in which urns are used to model credit default distributions \cite{Amerio}, allocation problems under uncertainty \cite{Arthur}, actuarial problems \cite{Daykin}, firms defaults \cite{Cihuesler}, risk and ambiguity \cite{Ghosh}, just to cite some papers.

\subsection{The idiosyncratic probability of default}\label{subsection2}

Consider $N$ failing systems (think of firms or bonds) divided into $k$ groups that, for simplicity, we assume to be homogeneous. We also hypothesize that, within each group, the FS are exchangeable in the sense of de Finetti \cite{Aldous}, that is to say that their joint probability is immune to permutations. This assumption is clearly weaker than the one of independence and identical distribution.\\
Each group consists of $n_{j}$ elements, such that $\sum_{j=1}^{k}n_{j}=N$. \\
Assume that every group is characterized by some sort of external information (qualitative or quantitative) about its reliability, i.e. about the reliability of its components. In other words we ask every group to possess a score $\gamma _{j}$, $j=1,...,k$, such that the set $G=\{\gamma
_{j}:j=1,...,k\}$ is a poset. This means that there exists a relation $\precsim\mathcal{\mathcal{}}$ on $G$ which is reflexive, antisymmetric and transitive or, formally,
\begin{gather*}
\gamma _{r}\precsim \gamma _{r} \\
\text{if }\gamma _{r}\precsim \gamma _{s}\text{ and }\gamma _{s}\precsim
\gamma _{r}\text{ then }\gamma _{r}=\gamma _{s} \\
\text{if }\gamma _{r}\precsim \gamma _{s}\text{ and }\gamma _{s}\precsim
\gamma _{t}\text{ then }\gamma _{r}\precsim \gamma _{t}.
\end{gather*}
Without any loss of generality, we will assume that the $k$ group are
completely ordered according to their ratings and, specifically, $\gamma
_{1}\succsim \gamma _{2}\succsim ...\succsim \gamma _{k-1}\succsim \gamma
_{k}$. We will read the relation $\succsim $ as ``better than" so, for
example, $\gamma _{1}\succsim \gamma _{2}$ means that group $1$ is ``better
than" group $2$ and, as a consequence of this, the $k-$th group is the worst
one, since it is characterized by the lowest reliability. In particular, as
signal of reliability we consider the idiosyncratic probability of default $D_{i}$ of the
different groups. In general - and this is neither a strong nor a ludicrous
assumption - we want that $D_{i}<D_{i+1}$ for $i=1,...,k$.\\
As far as the idiosyncratic probability of default of every group is concerned, we want to construct a mechanism that updates that probability every time a failure occurs in a given group. As said, our idea is to use Polya urns (see \cite{Johnson}) that are characterized by a simple but efficient reinforcement rule.\\
So, let $D_i(t)$ represent the idiosyncratic probability of default of groups $i$ at time $t$.\\
Assume that every group is associated with a Polya urn $U_i$, $i=1,...,k$ initially containing $w_i(0)\geq0$ white balls and $b_i(0)\geq0$ black balls. At time $t=1,2,...$, we sample a ball from urn $U_i$, look at its color and return it into the urn together with $s_i>0$ additional balls of the same color. This mechanism, called reinforcement, evidently modifies the composition of the urn, updating the probability of picking a certain color. If we repeat the sampling infinite times, we obtain an infinite sequence of $0-1$ Bernoulli random variables $\{X_i(t)\}$, where $X_i(t)=0$ if the sampled ball at time $t$ is black and $X_i(t)=1$ if white. The sequence $\{X_i(t)\}$ is exchangeable and it is called Polya sequence with parameters $(w_i(0),b_i(0),s_i)$.\\
For $t\geq0$ let $W_i(t)$ and $B_i(t)$ represent the number of white and black balls in urn $U_i$ at time $t$. It is easy to verify that
\begin{equation}
X_i(1)\sim Bern \left(\frac{w_i(0)}{w_i(0)+b_i(0)}\right)
\end{equation}
and, in general,
\begin{equation}
X_i(t+1)\sim Bern \left(\frac{W_i(t)}{W_i(t)+B_i(t)}\right),
\end{equation}
where $Bern(\eta)$ is the Bernoulli distribution of parameter $\eta$.\\
As far as the composition of the urn, the evolution rule is simply given by
\begin{equation}
(W_i(t+1),B_i(t+1))=
\begin{cases}
      (W_i(t)+s_i,B_i(t)) &\text{with probability $\frac{W_i(t)}{W_i(t)+B_i(t)}$}, \\
      (W_i(t),B_i(t)+s_i) &\text{with probability $\frac{B_i(t)}{W_i(t)+B_i(t)}$}.
\end{cases}
\end{equation}
\begin{proposition}\label{polyaseq}\\
 Let $\{X_i(t)\}$ be a Polya sequence with parameters $(w_i(0),b_i(0),s_i)$. Then: 
 \begin{enumerate}
 \item $\{X_i(t)\}$ is exchangeable and its de Finetti measure is a $Beta\left(\frac{w_i(0)}{s_i},\frac{b_i(0)}{s_i}\right)$;
 \item the proportion of white balls $Z_i(t)=\frac{W_i(t)}{W_i(t)+B_i(t)}$ converges with probability one to $p_i$.
 \end{enumerate}
\end{proposition}
The proof of this proposition is very well-known and we refer to \cite{Johnson} for a complete demonstration. Here below we only sketch the basic arguments.
\begin{proof}\\
First remember that $i=1,...,k$ are the different groups.\\
\textit{Point 1}: let $1\leq r \leq t$ and $(a_1,...,a_t)$ such that $a_l\in \{0,1\}$ and $\sum_{l=1}^t a_l=r$. Then
\begin{equation}
\begin{split}
P(X_i(1)=a_1,...,X_i(t)=a_t)&=\frac{\Gamma\left(\frac{w_i(0)}{s_1}+\frac{b_i(0)}{s_1}\right)}{\Gamma\left(\frac{w_i(0)}{s_i}\right)\Gamma\left(\frac{b_i(0)}{s_i}\right)}\frac{\Gamma\left(\frac{w_i(0)+r}{s_i}\right)\Gamma\left(\frac{b_i(0)+t-r}{s_i}\right)}{\Gamma\left(\frac{w_i(0)}{s_i}+\frac{b_i(0)}{s_i}+t\right)} \\&=\int_0^1\theta^r(1-\theta)^{t-r}\frac{\Gamma\left(\frac{w_i(0)}{s_1}+\frac{b_i(0)}{s_1}\right)}{\Gamma\left(\frac{w_i(0)}{s_i}\right)\Gamma\left(\frac{b_i(0)}{s_i}\right)}\theta^{\frac{w_i(0)}{s_i}-1}(1-\theta)^{\frac{b_i(0)}{s_i}-1}d\theta .
\end{split}
\end{equation}
For de Finetti's representation theorem this proves that the sequence is exchangeable. Furthermore, the unicity of the representation implies that the de Finetti measure of the sequence $\{X_i(t)\}$ is a $Beta\left(\frac{w_i(0)}{s_i},\frac{b_i(0)}{s_i}\right)$.\\
\textit{Point 2}: first we notice that $\{Z_i(t)\}$ is a bounded martingale, i.e.
\begin{equation}
\begin{split}
E[Z_i(t+1)|Z_i(1),...,Z_i(t)]&=\frac{W_i(t)+s_i}{W_i(t)+B_i(t)+s_i}\frac{W_i(t)}{W_i(t)+B_i(t)}\\ &+\frac{W_i(t)}{W_i(t)+B_i(t)+s_i}\frac{B_i(t)}{W_i(t)+B_i(t)}\\ &=\frac{W_i(t)}{W_i(t)+B_i(t)}=Z_i(t)
\end{split}
\end{equation}
Hence, for Doob's convergence theorem, the sequence $\{Z_i(t)\}$ converges almost surely to a random limit $Z_i(\infty)$.\\
From Point 1 and the law of large numbers we know that, for $t$ growing to infinity, the distribution of $t^{-1}\sum_{l=1}^t X_i(l)$ converges to a $Beta\left(\frac{w_i(0)}{s_i},\frac{b_i(0)}{s_i}\right)$. \\
For every $t\geq 1$
\begin{equation}
Z_i(t)=\frac{w_i(0)+s_i\sum_{l=1}^t X_i(t)}{w_i(0)+b_i(0)+ts_i},
\end{equation}
and for every $z\in[0,1]$
\begin{equation}
P(Z_i(t) \leq z)=P\left(t^{-1}\sum_{l=1}^t X_i(l)\leq z \left(\frac{w_i(0)}{ts_i}+\frac{b_i(0)}{ts_i}+1 \right)-\frac{w_i(0)}{ts_i}\right).
\end{equation}
Hence
\begin{equation}
\lim_{t \to \infty}P(Z_i(t) \leq z)=\int_0^z\frac{\Gamma\left(\frac{w_i(0)}{s_i}+\frac{b_i(0)}{s_i}\right)}{\Gamma\left(\frac{w_i(0)}{s_i}\right)\Gamma\left(\frac{b_i(0)}{s_i}\right)}\theta^{\frac{w_i(0)}{s_i}-1}(1-\theta)^{\frac{b_i(0)}{s_i}-1}d\theta
\end{equation}
\end{proof}
Without any loss of generality (it is just a rescaling), from now on we assume $w_i(0)\in[0,1]$ and $b_1(0)=1-w_i(0)$.\\ 
Thus, before observing $\{X_i(t)\}$ the probability of picking a white ball is equal to $Z_i(0)=w_i(0)$. Anyway, as soon as $n$ observation $X_i(1),...,X_i(t)$ are available, we update our beliefs about the probability of sampling white balls, i.e.
\begin{equation}\label{zn}
Z_i(t)=\frac{w_i(0)+s_i\sum_{j=1}^nX_i(j)}{1+ns_i}.
\end{equation}
Equation \ref{zn} shows why the Polya urn is one of the basic tools in Bayesian nonparametrics. In fact, Equation \ref{zn} is consistent with the Bayesian paradigm of prior specification, knowledge update thanks to empirical observations and posterior calculation. Our prior knowledge is given by the urn composition at time $0$. Then every time a white (or a black) ball is observed, our beliefs about the possibility of sampling white (or black) balls change and specifically increase, thanks to the reinforcement mechanism of the urn. Furthermore, given the initial composition and the updates, at every stage is possible to perform a prediction about the possibility of picking a given ball.\\
The informative contribution of every observation to the update process is given by the reinforcement quantity $s_i$, that is the number of balls added at every time step. It is clear that the relative contribution of an observation decreases with $n$, view that the more observation we have about process $\{X_i(t)\}$, the less we are ready to change our beliefs.\\
From Equation \ref{zn} and Proposition \ref{polyaseq} we also know that $Z_i(t)$ and $R_i(t)=t^{-1}\sum_{l=1}^tX_i(l)$, the rate of ones over $t$, have the same limit. Hence, for $t$ large enough, $R_i(t)$ is well approximated by the $Beta\left(\frac{w_i(0)}{s_i},\frac{b_i(0)}{s_i}\right) $ distribution.\\
Coming back to the construction of our model, and remembering the Polya urn scheme, let us now assume that, for every group $i$, $E[D_i(0)]=E[Z_i(0)]=w_i(0)$. In other words, we want to associate the probability of default to the sampling of white balls in urn $U_i$; one default in group $i$ corresponds to the extraction of one white ball from urn $U_i$. Moreover we make the hypothesis that $D_i(0)$ is distributed according to a $Beta(\frac{w_i(0)}{s_i},\frac{1-w_i(0)}{s_i})$. \\
Given $D_i(t)$, the default of element $j$ in group $i$ over the time horizon $t$ is distributed as a $Bern(D_i(t))$. If $n_i$ is the number of elements in group $i$, and for $j=1,...,n_i$, we let $\delta_i^{j}(t)$ represent the indicator function of the event ``default of element j in group i up to time t". In other words $\delta_i^{j}(t)=1$ if element j has defaulted at some point in the time interval [0,t] and $\delta_i^{j}(t)=0$ otherwise. We also let $\delta_i^1(t),\delta_i^2(t),...,\delta_i^{n_i}(t)$ be conditionally independent given $D_i(t)$. \\
Without any empirical observation and following our a priori knowledge, the probability of default at time $t$, $D_i(t)$,is clearly equal to $w_i(0)$. Anyway, it is very likely that, after a time $t$ has passed, we perfectly know how many elements of group $i$ have failed, thus we are ready to update our beliefs. In particular, following the Polya urn mechanism, we add $s_i$ white balls for every observed default and $s_i$ black balls for every surviving element. As a consequence of this $D_i(t)$ is distributed as a $Beta$ whose parameters are:
\begin{equation} \label{beta}
w^\ast_i(t)=\frac{w_i(0)+s_i\sum_{j=1}^{n_i}\delta_i^j(t)}{s_i}\; \; \text{and}\; \; b^\ast_i(t)=\frac{(1-w_i(0))+s_i\left(n_i-\sum_{j=1}^{n_i}\delta_i^j(t)\right)}{s_i}.
\end{equation} 
Moreover, thanks to Bayes theorem, it is straightforward to verify that the conditional distribution of $D_i(t)$ given $\delta_i^1(t),\delta_i^2(t),...,\delta_i^{n_i}(t)$ is a $Beta$ whose parameters are also expressed in Equation \ref{beta}, and that 
\begin{equation}\label{zeta}
E[D_i(t)|\delta_i^1(t),\delta_i^2(t),...,\delta_i^{n_i}(t)]=\frac{w_i(0)+s_i\sum_{j=1}^{n_i}\delta_i^j(t)}{1+s_in_i}.
\end{equation}
Thanks to this simple Polya-like urn scheme, we have thus modeled the idiosyncratic probability of default for every group. It is easy to see a clear relationship between the standard Polya urn and our adaptation to the idiosyncratic probability of default in group $i$, once we notice that, by forcing a little bit the notation, $X_i(t)=\sum_{j=1}^{n_i}\delta_i^j$.\\
Finally, it is important to notice that, being $w_i(0)$, $b_i(0)$ and $s_i$ generally different among groups, we are dealing with distinct $Beta$ distributions. Our use of the $Beta$ distributions for the idiosyncratic probabilities of default is consistent with several empirical and theoretical studies, as underlined in \cite{Amerio}.

\subsection{Modeling interaction: the systemic probability\newline
of default}\label{subsection3}

In our construction, the systemic probability of default accounts for the dependence among groups. In particular, once we know the idiosyncratic probabilities of default of the $k$ groups and we have ordered them from the most reliable to the least reliable one, we want the superior/best groups to a have a direct influence on the inferior/worst ones.\\
Let once again $D_1(t), D_2(t),..., D_k(t)$ be the idiosyncratic probabilities of default for the $k$ groups. It is evident that $D_i(t)\in(0,1)$. Now, define $D^{\ast }_i(t)$ as the total probability of default associated to group $i$ at time $t$, that's the ``sum" of the idiosyncratic and the systemic components. Having in mind the construction of \cite{Doksum} and \cite{Mulwal} for
neutral to the right processes, \cite{Ishwa} for stick-breaking priors, and \cite{Cirillo} for default models, we construct the probabilities of
failure of the $k$ groups as

\begin{gather} 
D^{\ast }_{1}(t)=D_{1}(t) \nonumber \\
D^{\ast }_{2}(t)=D^{\ast }_{1}(t)+(1-D^{\ast }_{1}(t))D_{2}(t)  \nonumber \\
\vdots \\ 
D^{\ast }_{k}(t)=D^{\ast }_{k-1}(t)+(1-D^{\ast }_{k-1}(t))D_{k}(t)=1-\prod_{i=1}^{k}\left( 1-D_{i}(t)\right) \nonumber.
\end{gather}

It is easy to verify that:

\begin{enumerate}
\item this construction respects all our assumptions, so that the better
groups of FS's show a lower probability of default;

\item the probabilities of default of the different groups are strictly
linked together by the means of the recursive scheme.
\end{enumerate}

So, thanks to this simple and rather intuitive iterative modeling, we obtain
the probabilities of default for the different risk groups and, for every
FS, we are able to say whether it is likely to fail.
\\ Please note that from now on we are omitting time $t$ not to perplex the notation; in other words, $D_i^\ast=D_i^\ast(t)$ and so on.

\begin{proposition} \label{nttr}\\
The process that governs the probabilities of failure $D^{\ast }_{i}$, $i=1,...,k$, is
a neutral to the right process.
\end{proposition}

\begin{proof}\\
First of all let $E_{i}=D^{\ast }_{i}-D^{\ast }_{i-1}$ with $i=1,...,k$ and $E_{1}=D^{\ast }_{1}=D_1$.
It is easy to verify that%
\begin{equation}
\left( E_{1},E_{2},...,E_{k}\right) \overset{d}{=}\left( D_{1},D_{2}\left( 1-D_{1}\right) ,...,D_{k}\prod_{i=1}^{k-1}\left( 1-D_{i}\right) \right) .  \label{equa}
\end{equation}%
This, as shown in remark $3.1b$ in Doksum \cite{Doksum}, assures that the process
governing $\left( D^{\ast }_{1},D^{\ast }_{2},...,D^{\ast }_{k}\right) $ is neutral to the right.
\end{proof}

Neutral to the right processes have been introduced by \cite{Doksum} and are widely used in Bayesian nonparametrics for survival analysis. For a complete introduction to this type of processes, we refer to the original paper by Doksum \cite{Doksum} and to other more recent works like \cite{Mulwal} and \cite{Walmul}. Here it is sufficient to state the following definition.

\begin{definition}[\protect Doksum \cite{Doksum}]\\
The random distribution function $F$ is said to be neutral to the right if
for each $h>1$ and $t_{1}<...<t_{h}$, there exist nonnegative independent
random variables $V_{1},...,V_{h}$ such that 
\begin{equation*}
\left( F\left( t_{1}\right) ,F\left( t_{2}\right) ,...,F\left( t_{h}\right)
\right) =_{\mathcal{L} }\left( V_{1},1-\left( 1-V_{1}\right) \left(
1-V_{2}\right) ,...,1-\prod_{i=1}^{h}\left( 1-V_{i}\right) \right) .
\end{equation*}%
The equations 
\begin{equation*}
F\left( t_{j}\right) =1-\prod_{i=1}^{j}\left( 1-V_{i}\right) \text{ \ }%
j=1,...,h
\end{equation*}%
yield%
\begin{equation*}
F\left( t_{j}\right) -F\left( t_{j-1}\right) =V_{j}\prod_{i=1}^{j-1}\left(
1-V_{i}\right)
\end{equation*}%
and%
\begin{equation*}
V_{j}=\frac{\left( F\left( t_{j}\right) -F\left( t_{j-1}\right) \right) }{%
\left( 1-F\left( t_{j-1}\right) \right) }\text{ \ }j=1,...,h\text{ and\ }%
t_{0}=-\infty .
\end{equation*}%
Thus \textquotedblleft $F$ is neutral to the right\textquotedblright\ mainly
means that the normalized increments%
\begin{equation*}
F\left( t_{1}\right) ,\frac{\left( F\left( t_{2}\right) -F\left(
t_{1}\right) \right) }{\left( 1-F\left( t_{1}\right) \right) },...,\frac{%
\left( F\left( t_{h}\right) -F\left( t_{h-1}\right) \right) }{\left(
1-F\left( t_{h-1}\right) \right) }
\end{equation*}%
are independent for all the $t_{1}<...<t_{h}$.
\end{definition}

The fact that we have modeled the idiosyncratic probabilities of default by the means of Polya urns has an interesting consequence.

\begin{proposition} \label{betas}\\
In particular the process that governs the probabilities of failure $D^\ast_{i}$, $i=1,...,k$, is
a beta-Stacy process with parameters $(w^\ast_1,b^\ast_1;w^\ast_2,b^\ast_2;...;w^\ast_k,b^\ast_k)$, where $w^\ast_i=w^\ast_i(t)$ and $b^\ast_i=b^\ast_i(t)$ are defined as in equation \ref{beta}.
\end{proposition}

\begin{proof}\\
Since the beta-Stacy process, as defined in \cite{Walmul}, is a special case
of neutral to the right process, when the independent variables are Beta
distributed, it is straightforward to prove the proposition. In fact, using
equation \ref{equa} we known that%
\begin{gather*}
E_{1}=D_{1}\\
E_{2}=D_{2}(1-D_{1}) \\
\vdots \\
E_{k}=D_{k}\prod_{i=1}^{k-1}\left( 1-D_{i}\right) .
\end{gather*}
An obvious consequence of this is that
\begin{eqnarray*}
E_{1} &\sim &\mathcal{BS}  (w^\ast _{1};b^\ast _{1};1) \\
E_{2}|E_{1} &\sim &\mathcal{BS}  (w^\ast_{2};b^\ast_{2};1-E_{1}) \\
&&\vdots \\
E_{k}|E_{k-1},...,E_{1} &\sim &\mathcal{BS} (w^\ast_{k};b^\ast_{k};1-\sum_{j=1}^{k-1}E_{j}),
\end{eqnarray*}%
where $\mathcal{BS}  (a;b;c)$ is the so-called beta-Stacy distribution introduced by \cite{Mihu}, whose density function is%
\begin{equation*}
\frac{1}{B(a,b)}x^{a-1}\frac{(c-x)^{b-1}}{c^{a+b-1}}I_{(0,c)}(x),
\end{equation*}%
with $B(a,b)$ representing the standard beta function.\\
Hence the final result immediately follows.
\end{proof}

We would like to stress that, while proposition \ref{betas} is strictly linked to the use of Polya urns and $Beta$ distributions to model the idiosyncratic probabilities of defaults, proposition \ref{nttr} holds in general. The only simple requirement is that the variables $D_i$ are i.i.d. and $D_i \in [0,1]$ for $i=1,...,k$. In other words, given the recursive construction, the urn chain, independently from the type of urn used, always generates neutral to the right processes. This result is quite useful in practice, since neutral to the right processes are conjugate, thus simplifying Bayesian prediction.

\begin{corollary}\\
If $b_i=\sum_{k>i}w_i$ (and $\sum_{i=1}^\infty w_i<\infty$) the beta-Stacy process that governs the probabilities of default simply becomes a Generalized Dirichlet distribution, $GD\left(w^\ast_1,b^\ast_1;w^\ast_2,b^\ast_2;...;w^\ast_k,b^\ast_k\right)$, as defined in \cite{Connomo}.
\end{corollary}

\begin{proof}\\
The proof is an application of Theorem 4.1 in Muliere and Walker \cite{Mulwal}.
\end{proof}
This last corollary, whose conditions are easily fulfilled, is very useful in applications. In fact, if the beta-Stacy process degenerates to a Generalized Dirichlet distribution, its parameters can be easily estimated using several existing computational techniques, from the expectation-maximization algorithm to the generalized method of moments  and other more advanced tools (see for example \cite{}).\\
Now, let $F_{i}$, $i=1,...,k,$ be the number of failures in the $i-$th group with 
$n_{i}$ elements. The (marginal) probability of having $F_{i}=f_{i}$
failures in the $i-$th group, with $0\leq f_{i} \leq n_i$, is equal to
\begin{equation}
P\left[ F_{i}=f_{i}\right] =E\left[ \binom{n_{i}}{f_{i}}(D_{i}^\ast)^{f_{i}}(1-D^\ast_{i})^{n_{i}-f_{i}}\right] =E\left[ P\left[F_{i}=f_{i}|D^\ast_{i-1}\right] \right] .
\end{equation}%
Then, using standard combinatorial considerations, the joint defaults of the first two groups can be computed as
\begin{eqnarray}
P\left[ F_{1}=f_{1},...,F_{k}=f_{k}\right] &=&E\left[ \binom{n_{1}}{f_{1}}
(D_{1}^\ast)^{f_{1}}(1-D_{1}^\ast)^{n_{1}-f_{1}}\binom{n_{2}}{f_{2}}
(D_{2}^\ast)^{f_{2}}(1-D_{2}^\ast)^{n_{2}-f_{2}}\cdots \right.  \nonumber\\
&&\cdots \left. \binom{n_{k}}{f_{k}}(D_{k}^\ast)^{f_{k}}(1-D_{k}^\ast)^{n_{k}-f_{k}}
\right] 
\end{eqnarray}
Given these details we have that, at time $t$, the number of default in
the first group (the best one) follows a beta-binomial
distribution (see \cite{Black}), or
\begin{equation}\label{n1}
P[F_1(t)=f_1]=\binom{n_1(t)}{f_1}\frac{B(w^\ast_1(t),b^\ast_1(t))}{B(w_1(0),b_1(0))},
\end{equation}%
where $n_1(t)$ is the number of failing systems in group $1$ at time $t$. This results comes directly from the use of Polya urns. \\
At this point, we can obtain the number of defaults in the first two groups, that is
\begin{equation}\label{n2}
\begin{split}
&P[F_1(t)=f_1,F_2(t)=f_2]=E[P[F_1(t)=f_1,F_2(t)=f_2|D_1^\ast(t),D_2^\ast(t)]]\\
&=\binom{n_1(t)}{f_1}\binom{n_2(t)}{f_2}E[D_1^{\ast}(t)^{f_1}D_2^{\ast}(t)^{f_2}(1-D_1^{\ast}(t))^{n_1(t)-f_1}(1-D_2^{\ast}(t))^{n_2(t)-f_2}]\\
&=\binom{n_1(t)}{f_1}\binom{n_2(t)}{f_2}E\left[D_1(t)^{f_1}(1-D_1(t))^{n_1(t)-f_1}(1-D_1(t))^{n_2(t)-f_2}(1-D_2(t))^{n_2(t)-f_2}\right.\times \\
&\times \left.\left(\sum_{i=0}^{f_2}\binom{f_2}{i}D_1(t)^{i}D_2(t)^{n_2(t)-i}(1-D_1(t))^{n_2(t)-i} \right)\right]=\binom{n_1(t)}{f_1}\binom{n_2(t)}{f_2} \times \\
&\times \sum_{i=0}^{f_2}\binom{f_2}{i}\frac{B(w_1^\ast(t)+i,b_1^\ast(t)+n_2(t)-i)}{B(w_1(0),b_1(0))}\frac{B(w_2^\ast(t)-i,b_2^\ast(t))}{B(w_2(0),b_2(0))}.
\end{split}
\end{equation}
It should be now clear that the joint probability of the number of defaults in the $k$ groups can be obtained continuing the iterative construction of equations \ref{n1} and \ref{n2} and standard combinatorial techniques, the result being a combination of beta-binomial distributions.\\

\section{An application to credit risk modeling}\label{section4}

The way in which international rating companies such as Moody's, S\&P and Fitch deal with credit risk and firms' defaults is definitely similar to the framework of interacting failing systems we have introduced. Think for example of $N$ firms divided into $k$ homogenous groups, that are ordered according to their financial reliability. For these reasons, we here present a simulation exercise related to firms' defaults and credit risk.\\
In order to use our model in applications, we need at least to know the quantities $D_i(0)$ for $i=1,...,k$, that is the idiosyncratic probabilities of default for every group at time 0. In general this should not be a problem, since the model can be initialized with historical data. \\
Here we propose a first application of our model to firms' defaults. In particular we show a simulation exercise using fictitious data. The suitability of the urn chain model for this kind of phenomena is definitely supported by empirical evidence (see \cite{Allen}).\\
Imagine we have 290 firms divided into three groups of reliability: A, B, C. Group A contains the 20 best firms on the market, whose probability of default is very low. Group C is made up of the riskiest firms and has 180 elements. Group B is the intermediate one and contains 90 firms. Even though simplified, this framework correctly reproduces the firms' classification structure induced by rating companies such as Moody's and Standard \& Poor's.\\
To initialize the model, i.e. to define the values $D_i(0)$ for $i=A,B,C$, we can use different methods. A first possibility is given by historical data: looking at the time series of defaults, we can try to estimate the three idiosyncratic probabilities of default and then run the urn chain. Another possibility is fully Bayesian and it is related to the prior knowledge of the researcher. In fact it is always possible to incorporate one's beliefs in $D_i(0)$ and then use the urn scheme as updating mechanism (see \cite{Cifa}). Finally, a third way for setting the model up is given by the credit spread approach of \cite{Duffi}, which is the solution we adopt here.\\
Following \cite{Duffi}, we consider that the value of the whole probability of default $D_i^\ast(1)$ for period $[0,1]$ is indirectly quoted on the market as the average one-year credit spread $\gamma_i(1)$ of group $i$. In other words, $\gamma_i(1)$ measures the average riskiness of a member of group $i$ as the difference between the zero coupon bond of that member and the risk-free interest rate on the market. The values of $\gamma_i(1)$ can be quite easily found on the market and are surely available to practitioners. In particular, in the financial literature (see \cite{Amerio}) it is common to assume
\begin{equation}\label{pred}
E[D^\ast_i(t)]=1-\exp(-t\gamma_i(t)).
\end{equation}
Hence $E[D^\ast_i(1)]=1-\exp(-\gamma_i(1))$. \\
Since the quantity we are interested is $E[D_i(0)]=w_i(0)$, that is the idiosyncratic probability of default at time 0, we can try to obtain it from $D^\ast_i(1)$, using  a linear term structure for credit spreads commonly used by traders (see for example \cite{Dai}). In fact, while it is not a problem to obtain actual 1-year default probabilities, it can be harder to have estimates for shorter periods. In general, traders do assume that, once we have fixed a time horizon, the risk of having one or more defaults before the expiry linearly decreases.\\
The use of an underlying linear term structure can be also a good way for making predictions about the 1-year probability of default every time some new information is available about the numbers of defaults in the different groups. In what follows we split every year in twelve months, but the same reasoning is available for weeks and even days.\\
In our construction we have decided to order the groups of failing systems using their idiosyncratic probability of default as rating measure. Obviously there are many possibilities for ordering and they all depend on the amount of available information. In our case, assuming that $\gamma_A(1)=0.02$, $\gamma_B(1)=0.06$ and $\gamma_C(1)=0.09$ are the one-year credit spreads for the three groups, and looking at the idiosyncratic probability of default we have that
\begin{equation}
\begin{split}
&E[D^\ast_A(1)]=E[D_A(1)]=1-\exp(-\gamma_A(1))=0.0198\\
&E[D^\ast_B(1)]=E[D^\ast_A(1)+(1-D^\ast_A(1))D_B(1)]=1-\exp(-\gamma_B(1))=0.0582\\
&E[D^\ast_C(1)]=E[D^\ast_B(1)+(1-D^\ast_B(1))D_C(1)]=1-\exp(-\gamma_C(1))=0.0861.
\end{split}
\end{equation}
Thanks to equation 3.1 we know that, for group A - the best one, the total probability of default is equal to the idiosyncratic component. For the other two groups, on the contrary, the total probability of default also includes a systemic component that accounts for the dependence between groups. As a consequence of this we have that
\begin{equation}
E[D_B(1)|D^\ast_A(1)]=\frac{1-\exp(-\gamma_B(1))-D^\ast_A(1)}{1-D^\ast_A(1)},
\end{equation}
and
\begin{equation}
E[D_C(1)|D^\ast_B(1)]=\frac{1-\exp(-\gamma_C(1))-D^\ast_B(1)}{1-D^\ast_B(1)}.
\end{equation}
Now let us hypothesize that, according to the linear term structure, the credit spread decreases of 0.0005 points every month. This value is fictitious but it is not very far from what practitioners generally assume (see \cite{Duffi} and \cite{Schon}). In this way, if $\gamma_A(1)=0.02$ then $\gamma_A(6/12)=0.02+0.0005*6=0.0230$, $\gamma_A(3/12)=0.0245$ and so on. Hence we have $\gamma_A(0)=0.0260$, $\gamma_B(0)=0.0660$ and $\gamma_C(0)=0.0960$.\\
Trivially for periods less than one year we have that for group A (and the same holds for B and C) the probability of failing between time $0$ and $i/12$ is such that
\begin{equation}
E[D^\ast_{A}(i/12)]=1-\exp(-\frac{i}{12}\gamma_{A}(i/12)).
\end{equation}
Now let us assume that for every group, at the end of every month, we observe the number of defaults expressed in the tables \ref{sim1}, \ref{sim2} and \ref{sim3}. For example in the first month we have 0 defaults for group A, 3 in group B and 25 in C.\\
We finally need to define the values of reinforcement $s_i$ for $i=A, B, C$. For simplicity we assume that the reinforcement is always the same for all the groups and equal to $0.05$ and $0.01$. Obviously one can define diverse reinforcement rules for the different groups.\\
Since $s_i$ represents the size of information update generated by every observation, we understand that a greater value is equivalent to a considerable reinforcement, while $0.01$ corresponds to a weaker update. Probably, as we show at the end of this section, the choice of $s_i$ is one of the most  sensible features of our model: from one side, it allows the researcher to incorporate an eventual a priori knowledge about the impacts of defaults; from the other, different values can produce quite different results in estimation (see tables \ref{sim1}, \ref{sim2} and \ref{sim3}). A good idea could be to calibrate $s_i$ such that the variability of the reinforced credit spreads is as close as possible to the historical variability of the spreads quoted on the market.\\
Using equations \ref{beta}, \ref{zeta} and 3.1 is now possible to perform our simulation, obtaining the results of tables \ref{sim1}, \ref{sim2} and \ref{sim3} (notice that ``uc" stands for urn chain).

\begin{table}[htb]
\begin{center}
\caption{Simulation results for group A}\label{sim1}
\begin{tabular}{|| l |ccccccc||}
	\hline
Month $(i)$ & 0&1&2&3&4&5&6\\	
\hline
$\gamma(i/12)$ &0.0260&0.0255&0.0250&0.0245&0.0240&0.0235&0.0230\\
	\hline
Defaults & 0&0&1&0&0&0&2\\
	\hline
$E[D^\ast_{uc,0.05}(i/12)]$&0.0257&0.0128&0.0314&0.0161&0.0083&0.0042&0.0535\\
	\hline
$E[D^\ast_{uc,0.01}(i/12)]$&  0.0257&0.0214&0.0262&0.0220&0.0185&0.0155&0.0298\\
	\hline
\hline
Month $(i)$ & 7&8&9&10&11&12&\\	
\hline
$\gamma(i/12)$ &0.0225&0.0220&0.0215&0.0210&0.0205&0.0200&\\
	\hline
Defaults &1&0&0&0&0&0&\\
	\hline
$E[D^\ast_{uc,0.05}(i/12)]$&0.0559&0.0311&0.0173&0.0096&0.0053&0.0030&\\
	\hline
$E[D^\ast_{uc,0.01}(i/12)]$&0.0341&0.0294&0.0253&0.0218&0.0188&0.0162&\\
	\hline
\hline
\end{tabular}
\end{center}
\end{table}

\begin{table}[htb]
\begin{center}
\caption{Simulation results for group B}\label{sim2}
\begin{tabular}{|| l |ccccccc||}
	\hline
Month $(i)$ & 0&1&2&3&4&5&6\\
	\hline
$\gamma(i/12)$ &0.0660&0.0655&0.0650&0.0645&0.0640&0.0635&0.0630\\
	\hline
Defaults & 0&3&1&0&4&5&8\\
	\hline
$E[D^\ast_{uc,0.05}(i/12)]$ &0.0639&0.0468&0.0467&0.0190&0.0462&0.0605&0.1426\\
	\hline
$E[D^\ast_{uc,0.01}(i/12)]$&0.0639&0.0570&0.0503&0.0350&0.0466&0.0581&0.0974\\
	\hline \hline
Month $(i)$& 7&8&9&10&11&12 &\\
	\hline
$\gamma(i/12)$ &0.0625&0.0620&0.0615&0.0610&0.0605&0.0600&\\
	\hline
Defaults &9&5&5&4&0&2&\\
	\hline
$E[D^\ast_{uc,0.05}(i/12)]$ &0.1714&0.1212&0.1072&0.0921&0.0304&0.0408&\\
	\hline
$E[D^\ast_{uc,0.01}(i/12)]$&0.1253&0.1170&0.1135&0.1069&0.0773&0.0698&\\
	\hline \hline
\end{tabular}
\end{center}
\end{table}

\begin{table}[htb]
\begin{center}
\caption{Simulation results for group C}\label{sim3}
\begin{tabular}{|| l |ccccccc||}
	\hline
Month $(i)$ & 0&1&2&3&4&5&6\\
	\hline
$\gamma(i/12)$ & 0.0960&0.0955&0.0950&0.0945&0.0940&0.0935&0.0930\\
	\hline
Defaults & 0&25&19&9&14&10&24\\
	\hline
$E[D^\ast_{uc,0.05}(i/12)]$ & 0.0915&0.1688&0.1641&0.0911&0.1466&0.1460&0.3225\\
	\hline
$E[D^\ast_{uc,0.01}(i/12)]$ & 0.0915&0.1512&0.1583&0.1183&0.1417&0.1464&0.2458\\
	\hline \hline
Month $(i)$ &7&8&9&10&11&12& \\
	\hline
$\gamma(i/12)$ &0.0925&0.0920&0.0915&0.0910&0.0905&0.090&\\
	\hline
Defaults &15&14&9&9&9&7& \\
	\hline
$E[D^\ast_{uc,0.05}(i/12)]$ & 0.3321&0.3083&0.2763&0.2824&0.2764&0.3101& \\
	\hline
$E[D^\ast_{uc,0.01}(i/12)]$ &0.2789&0.2869&0.2805&0.2832&0.2782&0.2874&\\
	\hline
\hline
\end{tabular}
\end{center}
\end{table}

Every table is devoted to one of the three groups A, B, C. Each table contains the following information: the months from 0 (present time) to 12 (one year); the values of the credit spread at one year $\lambda_i$ in the different months according to the linear spread term structure; the number of defaults in the different months; the estimated probability of default according to our urn chain model with two different values for the reinforcement quantity $s_i$, that is $0.01$ and $0.05$. \\
As expected our model gives estimates of the probability of default that are clearly different form the basic predictor based on equation \ref{pred}. In particular, it is evident how the numbers of defaults in the different months have a clear impact on the probability of default.\\ 
This information about defaults would be probably neglected in a standard approach without reinforcement or, in the best case, all the update would be performed at the end of the 1-year period, 
when all defaults have happened, and only as a basis for the next period of interest, with a clear temporal delay. Unfortunately this eventual all-in-one update is not really useful at all, since it can make the probability of default increase for all the second year, even if during the second year no default actually happens. In other words, a practitioner could base his/her evaluations on the basis of out-of-date facts. Our model is instead continuously updated, always producing updated estimates. \\
For every group, the role of the urn reinforcement mechanism is really clear if we compare the expected probability of default with $s_i=0.05$ and $s_i=0.01$. The greater is the reinforcement the greater are the fluctuations of the expected probability of default after every update. For example $s_i=0.05$ seems to be a quite high value for the updating process: if we consider group A and periods $5$ and $6$, we see that with $s_A=0.05$ the probability of default jumps from $0.0042$ to $0.0535$, indicating an excessive sensitivity of the model to defaults. For $s_A=0.01$, on the contrary, the jump is more contained, from $0.0155$ to $0.0298$, suggesting a more plausible variation. \\
Comparing the three tables, we can finally notice that, thanks to the urn chain mechanism, every time the probability of default increases (decrease) in group A, the dependence structure makes the probabilities of default of the inferior groups increase too. In other words, avoiding the assumption of independent defaults (both within and between groups), we have tried to overcome one of the weakest points of standard credit risk models (see for example the CR+ model in \cite{Credi}).\\

\section{Conclusion}
In this paper we have shown a first (Bayesian) nonparametric model for studying failing systems. \\
In detail, we consider failing systems divided into homogeneous groups of different reliability and we assume that these groups can be ordered according to some sort of external information. The elements within each group are assumed to be exchangeable.\\
We hypothesize that the probability of default of every failing system is given by the sum of two different components: an idiosyncratic probability of default related to the group to which the FS belongs and a systemic probability of defaults that account for the dependence between groups. For the first probability we make use of Polya urns that allows for a Bayesian nonparametric modeling based on information update, while for the second probability we construct an urn chain whose structure is very close to the one used by \cite{Doksum} for constructing neutral to the right processes.\\
We have also proposed a possible application of our model; in particular, we show a simulation experiment related to credit risk modeling. At this point it could be worth to apply the model to actual problems and data and to compare it with some benchmark.\\
As far as the evolution of the model is concerned it could be interesting to substitute the simple Polya mechanism with a more advanced scheme. An idea could be to use reinforced urn processes (see \cite{Walmul}) to model the idiosyncratic probabilities of defaults and then to combine them using the same neutral to the right construction. We believe that the general process governing the probability of default would still be a beta-Stacy process, but further analysis is needed.\\
Another research line could be to introduce the possibility of transitions from one group to the other, in order to model down- and upgrading of firms in reliability classes. This could be done by introducing random reinforcement rules.\\
Finally, we would like to stress that the choice of Polya urns to model the idiosyncratic probability of default is only due to the desire of obtaining closed-form results for the number of defaults. In reality, several other urn schemes could be used, while maintaining the general neutral to the right structure of the model (see \cite{Cirillo} for more details).\\
\\
\textbf{Acknowledgements:} The authors would like to thank the editor and the referees for their valuable suggestions. This work has been partly supported by the Swiss National Science Foundation, to which the first two authors are particularly grateful.

\end{document}